%% file: main_full.tex
\documentclass{article}
\usepackage[utf8]{inputenc}
\usepackage{fullpage}

\usepackage{xspace}

\usepackage[numbers, sort]{natbib}

\usepackage{amsthm,amsmath,amssymb}

\usepackage[dvipsnames,usenames]{xcolor}
\usepackage[colorlinks=true,pdfpagemode=UseNone,urlcolor=RoyalBlue,linkcolor=RoyalBlue,citecolor=OliveGreen,pdfstartview=FitH]{hyperref}
\usepackage{cleveref}
\usepackage{prettyref}

\usepackage[semibold,type1]{libertine} 
\usepackage[T1]{fontenc} 
\usepackage{textcomp} 
\usepackage[varqu,varl]{inconsolata}
\usepackage[scr=rsfso]{mathalfa}
\usepackage{bm}
\usepackage{graphicx}
\usepackage{indentfirst}
\usepackage{comment}
\usepackage{thm-restate}

\newtheorem{theorem}{Theorem}
\newtheorem{lemma}[theorem]{Lemma}
\newtheorem{fact}[theorem]{Fact}
\newtheorem{corollary}[theorem]{Corollary}
\newtheorem{definition}[theorem]{Definition}

\renewcommand{\mid}{\;\middle\vert\;}

\usepackage{algorithm, algorithmicx, algpseudocode}

\def\*#1{\mathbf{#1}} \def\+#1{\mathcal{#1}} \def\-#1{\mathrm{#1}} \def\^#1{\mathbb{#1}} \def\$#1{\mathtt{#1}}

\usepackage{xifthen}

\renewcommand{\Pr}[2][]{ \ifthenelse{\isempty{#1}}
  {\mathbf{Pr}\left[#2\right]} {\mathbf{Pr}_{#1}\left[#2\right]} }
\newcommand{\E}[2][]{ \ifthenelse{\isempty{#1}}
  {\mathbf{E}\left[#2\right]}
  {\mathop{\mathbf{E}}_{#1}\left[#2\right]} }
\newcommand{\Var}[2][]{ \ifthenelse{\isempty{#1}}
  {\mathbf{Var}\left[#2\right]}
  {\mathop{\mathbf{Var}}_{#1}\left[#2\right]} }

\usepackage{diagbox}
\usepackage{float}
\newfloat{figtab}{htb}{fgtb}
\makeatletter
  \newcommand\figcaption{\def\@captype{figure}\caption}
  \newcommand\tabcaption{\def\@captype{table}\caption}
\makeatother

\usepackage{expl3}
\ExplSyntaxOn
\newcommand\latinabbrev[1]{
  \peek_meaning:NTF . {
    #1\@}%
  { \peek_catcode:NTF a {
      #1.\@ }%
    {#1.\@}}}
\ExplSyntaxOff

\def\ie{\latinabbrev{i.e}}

\newcommand{\ALG}{\textup{\sf ALG}\xspace}
\newcommand{\OPT}{\textup{\sf OPT}\xspace}
\newcommand{\Batch}{\textup{\sf Batch-Rent}\xspace}
\newcommand{\Oracle}{\textup{\sf Oracle}\xspace}
\newcommand{\CoOPT}{\textup{\sf coOPT}\xspace}
\newcommand{\EDF}{\textup{\sf EDF}\xspace}
\newcommand{\Semi}{\textup{\sf SemiOnline}\xspace}
\newcommand{\calJ}{\mathcal{J}}

\title{Improved Algorithms for Online Rent Minimization Problem Under Unit-Size Jobs}
\author{Enze Sun \thanks{The University of Hong Kong. Email: sunenze@connect.hku.hk.} \and
Zonghan Yang \thanks{Shanghai Jiao Tong University. Email: fstqwq@sjtu.edu.cn.} \and
Yuhao Zhang \thanks{Shanghai Jiao Tong University. Email: zhang\_yuhao@sjtu.edu.cn.}
}
\date {}

\renewcommand{\subparagraph}{\paragraph}

\begin{document}

\maketitle

\input{abstract}
\input{introduction}
\input{preliminary}

\input{oracle_based_framework}

\input{semi.tex}


\input{onlinealg}
\input{lambda_full}

\input{conclusion}

\section*{Acknowledgement}

The authors would like to thank Minming Li, Pinyan Lu, and Biaoshuai Tao for many insightful discussions on the research topic of this paper. Yuhao Zhang is supported by National Natural Science Foundation of China, Grant No. 62102251.

\bibliographystyle{plainnat}
\bibliography{ref}

\end{document}

%% file: abstract.tex
\begin{abstract}
    We consider the Online Rent Minimization problem, where online jobs with release times, deadlines, and processing times must be scheduled on machines that can be rented for a fixed length period of $T$. The objective is to minimize the number of machine rents. This problem generalizes the Online Machine Minimization problem where machines can be rented for an infinite period, and both problems have an asymptotically optimal competitive ratio of $O(\log(p_{\max}/p_{\min}))$ for general processing times, where $p_{\max}$ and $p_{\min}$ are the maximum and minimum processing times respectively. However, for small values of $p_{\max}/p_{\min}$, a better competitive ratio can be achieved by assuming unit-size jobs. Under this assumption, Devanur et al. (2014) gave an optimal $e$-competitive algorithm for Online Machine Minimization, and Chen and Zhang (2022) gave a $(3e+7)\approx 15.16$-competitive algorithm for Online Rent Minimization. In this paper, we significantly improve the competitive ratio of the Online Rent Minimization problem under unit size to $6$, by using a clean oracle-based online algorithm framework. 
\end{abstract}

%% file: introduction.tex
\section{Introduction}

\emph{Machine Minimization} is a classical scheduling problem in combinatorial optimization. We are given $n$ jobs with release time and deadline to schedule. Each job $j$ has a length $p_j$ and must be assigned to a machine for $p_j$ units of time between its release time $r_j$ and its deadline $d_j$. However, in many practical scenarios, such as cloud computing, we may not need to buy the machines but only rent them for a fixed period of time. This motivates the \emph{Rent Minimization} problem, introduced by Saha~\cite{saha_renting_2013}. In this problem, we are given a constant $T$, which represents the duration of a machine rent. The goal is to minimize the number of rents we make to process all jobs within their deadlines.

Another related formulation, inspired by nuclear weapon testing, is the \emph{Calibration} problem, proposed by Bender et al.~\cite{bender_efficient_2013}. In this problem, we are given $m$ machines and a set of jobs that must be completed feasibly. However, before using a machine, we need to \emph{calibrate} it. Each calibration, similar to a rent, activates the machine for a time period of $T$. The goal is to minimize the number of calibrations to process all jobs on time. The main difference between the Calibration and the Rent Minimization problems is that the former restricts us to have at most $m$ machines working in parallel at any given time, while the latter does not have such a constraint. Therefore, the Rent Minimization problem can be regarded as a special case of the Calibration problem when $m=\infty$.

On the other hand, in the cloud rental scenario and many other practical applications, the computing requests usually increase over time and can be modeled as online released jobs. Therefore, we investigate the problem in an online setting. We do not have any prior knowledge about the jobs before their release time, and need to schedule jobs and rent machines online and irrevocably over time. The goal is to minimize the total number of \emph{rents} for scheduling all jobs. Note that the online generalization is also studied in the Calibration problem by Chen and Zhang~\cite{chen_online_2022}. To ensure that online algorithms can schedule all jobs, they also assume $m=\infty$ in their model, which coincides with the Online Rent Minimization model.

\subparagraph*{Why consider unit-size jobs?} Saha~\cite{saha_renting_2013} proposes an $O(\log{(p_{\max}/p_{\min})})$-competitive algorithm for the Online Machine Minimization problem. By paying a constant factor, it can be extended to an $O(\log{(p_{\max}/p_{\min}))}$-competitive algorithm for the Online Rent Minimization problem. ($p_{\max}$ and $p_{\min}$ are the longest and shortest processing time among all jobs.), which was proved to be the best competitive ratio asymptotically. However, in many real-world applications, one company usually receives similar length requests, so the ratio between $p_{\max}$ and $p_{\min}$ may not be too large; and it is worthwhile to reduce the constant factor of the  ratio when $p_{\max}/p_{\min}$ is small. To this end, we focus on the special case of unit-size jobs (i.e., all $p_j=1$).  Note that by partitioning jobs by their length into $\log(p_{\max}/p_{\min})$ groups, the $\alpha$-competitive unit-size algorithm can be extended to a roughly $(3\alpha \log{(p_{\max}/p_{\min})})$-competitive algorithm in the general case.

The unit-size special case has also been considered in the Online Machine Minimization problem~\cite{chen_online_2022, devanur_online_2014, DBLP:conf/isaac/KaoCRW12}. Devanur et al.~\cite{devanur_online_2014} present an $e$-competitive algorithm for the Online Machine Minimization problem under unit-size jobs (though earlier work by Bansal et al. ~\cite{10.1145/1206035.1206038} implies the same result), and it is the optimal ratio among all deterministic algorithms. Current best lower bound of the online renting problem under unit-size jobs is also $e$ since it is a generalized model. On the positive side, Chen and Zhang~\cite{chen_online_2022} study the online renting problem under unit-size jobs. They improve the implicit constant ratio by Saha~\cite{saha_renting_2013} to $3e+7\approx 15.16$. In their algorithm, jobs are distinguished by whether they are long or short based on the length of their time window (i.e., $d_j-r_j$) and are handled separately. Our paper significantly improves the competitive ratio to $6$ with a cleaner oracle-based algorithm without identifying whether jobs are short or long.

\begin{theorem}
\label{thm:main}
There exists an efficient $ 6$-competitive algorithm for the online renting problems under unit-size jobs. 
\end{theorem}

\subparagraph*{Our techniques.} In the work of Chen and Zhang~\cite{chen_online_2022}, they rent machines for long and short jobs separately; as a result, their final competitive ratio is the sum of two cases, which makes the ratio large. 
The technical reason behind this result is that they use the $e$-competitive Online Machine Minimization algorithm by Devanur et al.~\cite{devanur_online_2014} as a black box, which is only suitable for \emph{short} jobs. (Roughly speaking, it is because we can view $T=\infty$ when jobs are short.) 
Therefore, they must use another approach to handle long jobs.

In our paper, we formalize and extend the Online Machine Minimization algorithm to an oracle-based framework, instead of using the algorithm as a black box. The oracle-based framework uses an offline algorithm to guide our online decision. Note that the Online Machine Minimization algorithm also uses an efficient offline optimal algorithm as an oracle. However, we do not know a polynomial offline Rent Minimization algorithm for unit-size jobs. 
The main algorithmic novelty is that we find an efficient substitute for the optimal algorithm to act as a bridge between online decisions and the optimal offline solution. 
The oracle is a kind of optimal augmentation algorithm. It is allowed to use a rent length of $3T$, and the rent number is at most $\OPT$ with rent length $T$. It also satisfies some online monotone properties so that we can control the cost of the online algorithm. 
Finally, we prove that by following the offline oracle and paying a factor of $6$ online, we can recover the same ability for scheduling jobs as the offline oracle. This concludes the competitive ratio of $6$.

\subparagraph*{Extension to the model with delay.} Chen and Zhang~\cite{chen_online_2022} also raise a perspective that the operation \emph{rent} (a.k.a. \emph{calibration} in their paper) needs a non-negative time $\lambda$ to finish. We call it \emph{Online Rent Minimization with Delay}. They propose an $\left(3(e+1)\lambda + 3e+7\right)\approx (11.15\lambda+15.16)$-competitive algorithm. We use a black box reduction to extend the algorithm in \Cref{thm:main} and improve the ratio to $6 (\lambda+1)$.

\begin{restatable}{theorem}{timecrit}
\label{thm:timecritical}
As a corollary of \Cref{thm:main}, there exists an efficient $6(\lambda+1)$-competitive algorithm when we need $\lambda$ time to finish each rent. 
\end{restatable}

\subparagraph*{Other related works.} Offline Machine Minimization is a well-studied and classic model. Garey and Johnson~\cite{DBLP:books/fm/GareyJ79} shows that it is NP-hard. On the algorithm side, Raghavan and Thompson~\cite{DBLP:journals/combinatorica/RaghavanT87} propose an $O(\frac{\log n}{\log \log n})$-approximation algorithm. Later, the ratio has been improved to $O(\sqrt{\frac{\log n}{\log \log n})}$ by Chuzhoy et al.~\cite{chuzhoy_machine_2004}. Whether there exists a constant approximation ratio is still open. Moreover, several special cases are also discussed. Cieliebak et al.~\cite{DBLP:conf/ifipTCS/CieliebakEHWW04} focus on the case that each job's active time $(d_j-r_j)$ is small.  Yu and Zhang~\cite{DBLP:journals/orl/YuZ09} achieve a ratio of 2 in the equal release time case and a ratio of 6 in the equal processing time case. 

Scheduling to minimize the number of calibrations is proposed by Bender et al.~\cite{bender_efficient_2013}. The general case is NP-hard even for checking feasibility. Under unit-size jobs, Bender et al. give a $2$-approximation algorithm; later, Chen et al.~\cite{chen_approximation_2019} give the first PTAS algorithm. However, it is worth noting that whether the unit-size special case is polynomially solvable is still open. Moreover, Angel et al.~\cite{angel_complexity_2017} introduce the concept of \textit{delay}, which means that each calibration requires $\lambda$ time to finish. They study the delay setting on the one-machine special case of the offline calibration problem and show that it is polynomially solvable.


%% file: preliminary.tex
\section{Preliminaries}
We first define the models and introduce the basic notations. 

\subparagraph*{Rent Minimization.} We have a set of jobs $\mathcal{J} = \{1, \cdots, n\}$ and a fixed rent length $T$. For job $j \in \mathcal J$, it has a release time $r_j$, a deadline $d_j$, and a unit processing time $p_j = 1$.
Each job should be assigned to one active machine at an integer time unit $[t,t+1)$, where $t$ is an integer such that $r_j\leq t \leq d_j-1$. We can rent a machine at any integer time point $t$. Then we will have an active machine during $[t,t+T)$. The objective is to minimize the number of machine rents to process all jobs in $\mathcal{J}$.

\subparagraph*{Online Rent Minimization.} In the online version, all jobs are released online, and they become \emph{visible} at their release time. On the other hand, we need to make rent decisions and assign jobs online irrevocably. In particular, at an integer time point $t$, we have:
\begin{itemize}
    \item Jobs with release time equal to $t$ become visible.  
    \item We can decide to rent a machine at $t$ or any time after that. 
    \item We can schedule jobs on active machines during the time unit $[t,t+1)$ irrevocably. 
\end{itemize}

\subparagraph*{Notations on Rent Set.}
We use a multiset $I = \{[s_1,c_1), [s_2, c_2), \cdots [s_i,c_i), \cdots \}$ to denote a set of rents, where the $i$-th rent interval starts at $s_i$ and is active in $[s_i,c_i)$. 
Note that $c_i$ always equals $s_i+T$ when the rent length is $T$; however, we use the general notation for future reference.

Focusing on the time unit $[t,t+1)$, the number of \emph{active units} $A_I (t)$ is defined as the number of active machines at time $t$, which means that we can schedule at most $A_I (t)$ jobs at time $t$. For a given rent set $I$, we have
$
A_I (t) = \left\vert \left\{[s_i, c_i) \in I \mid t \in [s_i, c_i)\right\}\right\vert.
$
We also extend the notation for intervals, such that
$
A_{I}(r^*,d^*) = \sum_{t = r^*}^{d^* - 1} A_I(t)
$
is the number of active units during the time interval $[r^*,d^*)$.

\subparagraph*{Feasibility of Rent Set.} We call a rent set $I$ \emph{feasible} for $J$, if we can schedule all jobs in $J$ on $I$. We introduce a lemma based on Hall's Theorem to check whether $I$ is feasible.

For a given jobs set $J$, we define
$
    J(r^*,d^*) = \left\{j \in J | r^* \leq r_j < d_j \leq d^* \right\}
$
to represent the jobs that must be assigned inside the interval $[r^*,d^*)$. 

\begin{lemma}[Feasibility]
    \label{lem:fea}
$I$ is feasible for $J$ iff. 
    $
    \forall r^* \leq d^* \in \mathbb N,~ A_I(r^*,d^*) \geq |J(r^*,d^*)|~.
    $
\end{lemma}
\begin{proof}
    For any fixed $r^*$ and $d^*$, the sum of active units provided by $I$ is $A_I(r^*,d^*)$. Each job released and due between this period must be scheduled on these time slots. If there exists a pair of $r^*$ and $d^*$ such that $A_I(r^*,d^*) < |J(r^*,d^*)|$, there is no feasible assignment because of the pigeonhole principle.
    On the other hand, if the inequality holds for all $r^*$ and $ d^*$, we can view it as a bipartite matching between jobs and active units. There is a feasible assignment by Hall's Theorem. 
\end{proof}
\subparagraph*{An Efficient Checker and Scheduler: Earliest Deadline First (EDF).}
\textsf{\textbf{E}arliest \textbf{D}eadline \textbf{F}irst} is a greedy algorithm that can find a feasible assignment for $J$ on $I$ if and only if $I$ is feasible for $J$. When we call $\EDF(J, I)$, we scan time units from early to late, and assign the released job with the earliest deadline to a free active machine at the current time unit. If a job with deadline $d$ cannot find a free active machine at the time unit $[d-1,d)$, $\EDF(J, I)$ fails, and we call $d$ the \emph{fail time} of $\EDF(J, I)$. Otherwise, $\EDF(J,I)$ succeeds. Bender et al. \cite{bender_efficient_2013} has already proved that $\EDF$ can check the feasibility. It is also worth noting that $\EDF$ can be efficiently implemented in $O(n \log n)$ by using a heap, instead of going through all integer time points directly. 

\begin{lemma} [\cite{bender_efficient_2013}]
    \label{lem:edf}
    $\EDF(J,I)$ succeeds, \ie, it can find a feasible schedule for $J$ on $I$, if and only if $I$ is feasible for $J$. 
\end{lemma}

\subparagraph*{Using EDF online.} Note that we only make comparisons between released jobs. Therefore, the EDF algorithm can be simulated online : we only need to find a feasible rent set $I$, and then EDF can automatically find a feasible assignment online.

%% file: oracle_based_framework.tex
\section{Oracle-based Online Algorithm Framework}

Moving towards online algorithms, one natural way is to use an offline algorithm as an \emph{oracle} to suggest the actions of online algorithms. We keep track of this offline algorithm and make corresponding online decisions when the offline algorithm changes along with the online jobs release. Whenever the offline algorithm increases by one at the moment $t$ because of the change of the job set, the online algorithm performs one \Batch at this time $t$, which is a fixed rent scheme that contains $\Gamma$ machines. Intuitively, we use these $\Gamma$ machines to catch up with the one increment of the offline oracle. It is worth noting that the $e$-competitive algorithm for Online Machine Minimization follows this approach~\cite{devanur_online_2014}. 

In our case, compared to Machine Minimization, we have two main differences in our oracle-based framework. The first difference is the job set we input to the oracle. Because of the online fashion, the most natural way is to input the set of all released jobs to the oracle. In Machine Minimization, it works because $T=\infty$ and earlier rents are always more powerful; however, this is not true in Rent Minimization. Indeed, too early rents may cause trouble in Rent Minimization.
Intuitively, we are only allowed to make online rent when the offline oracle reports an increment if we want to bound the competitive ratio in the framework. Consider the case where $T=10$ and two jobs will be due at $100$ with release time $0$ and $90$. If we report the first job to the offline oracle at $0$, the offline oracle will return one new rent interval. Following the oracle-based framework, we will make $\Gamma$ online rent intervals at $0$. However, we still need to make more rent intervals at $90$, while the offline oracle may not increase. To this end, we use the job set $J_t$ as the job set we input to the oracle at time $t$. A job $j$ is in $J_t$ if it satisfies the following two properties:
\begin{enumerate}
    \item[1)] It is \emph{visible} at $t$, i.e., $r_j\leq t$. 
    \item[2)] It is \emph{emergent} at (or before) $t$, i.e., $d_j \leq t + T$.
\end{enumerate}

Another difference is an augmentation factor $\alpha$. The oracle is allowed to have $\alpha T$ active time for each rent. Since the existence of a polynomial time optimal offline algorithm for Rent Minimization is still unknown, this factor allows us to find an efficient substitute. We use $\Oracle_\alpha (J, T)$ (instead of $\Oracle (J, \alpha T)$, since the target rent length is still $T$) to denote an oracle with augmentation factor $\alpha$. The framework is formalized in \Cref{alg:oracleframework}.

\begin{algorithm}[htbp]
  \caption{Oracle-based Online Algorithm Framework}
  \label{alg:oracleframework}
  \begin{algorithmic}
  \Procedure{{\sf OracleBasedOnline}}{$t$: time, $J$: known jobs, $T$: length of rent}
  \State $\Delta_t = |\Oracle_\alpha(J_t, T)| - |\Oracle_\alpha(J_{t - 1}, T)|$ \Comment{$\alpha$ is a positive integer}
  \State Perform $\Delta_t$(if $\Delta_t>0$) \Batch operations at $t$, consisting of $\Gamma$ machines that start at or after $t$.
  \State {\bf schedule} jobs at $t$ following $\EDF(J,I)$, where $I$ is the current online rent set. 
  \EndProcedure
  \end{algorithmic}
\end{algorithm}




Then, we discuss how this framework helps us control the number of rents made by the online algorithm. First, as a substitute for the optimal offline algorithm, $\Oracle_{\alpha}$ needs to maintain some properties similar to those of the optimal offline algorithm. Second, \Batch should support the online algorithm to be as powerful as the offline oracle in scheduling all released jobs.
We integrate and formalize these messages in the following lemma. 
\begin{lemma}
    \label{lem:oracle}
    Let $\OPT(J, T)$ be the number of rents used by the optimal offline algorithm to schedule the job set $J$,  
    \Cref{alg:oracleframework} is $\Gamma$-competitive if these three properties are guaranteed: 
    \begin{enumerate}
        \item[1)] For any job set $J$, $|\Oracle_{\alpha}(J, T)| \leq \OPT(J,T)$;
        \item[2)] The offline oracle is online monotone: $|\Oracle_{\alpha}(J_{t_1}, T)| \leq |\Oracle_{\alpha}(J_{t_2}, T)|$ if $t_1 \leq  t_2$;
        \item[3)] \Cref{alg:oracleframework} is feasible for scheduling all online released jobs. 
    \end{enumerate}
\end{lemma}

\begin{proof}
The online algorithm makes $\sum _{\Delta_t>0} \Delta_t$ rent batches, which is exactly $|\Oracle_{\alpha} (\mathcal{J}, T)|$ by property 2) and is not greater than $|\OPT(\mathcal{J}, T)|$ by property 1). Also, the output satisfies the feasibility requirement by property 3). Therefore, \Cref{alg:oracleframework} is $\Gamma$-competitive.
\end{proof}

\subparagraph*{The $e$-competitive algorithm for Online Machine Minimization.} We can use the framework to understand the $e$-competitive Online Machine Minimization algorithm.
\begin{itemize}
    \item \Oracle is the optimal offline algorithm, and we set $\alpha=1$. The monotonicity directly comes from optimality. 
    \item $J_t$ is the set of visible jobs at $t$ because all jobs are emergent. 
    \item Each $\Batch$ contains $e$ new machines in average; for simplicity, we omit any rounding issues related to $e$.
\end{itemize}

\subparagraph*{Choice of the oracle.} 
Recall that we do not have an efficient optimal algorithm for Rent Minimization currently. It remains to find a suitable substitute that also uses a small number of rents (property 1). One candidate algorithm may be the Lazy-Binning algorithm by Bender et al.~\cite{bender_efficient_2013}, which requires an augmentation factor of $2$ to satisfy property 1). However, Lazy-Binning algorithm, as well as other relatively simple 2 approximation algorithms we come up with, cannot guarantee monotonicity. This will make us fail to bound the competitive ratio of the online algorithm. 
In the next section, we introduce our oracle with an augmentation factor of $3$, called the semi-online algorithm, which provides all the properties we need in \Cref{lem:oracle}.

%% file: semi.tex
\section{The Semi-Online Algorithm}
\label{sec:semi}
In this section, we introduce the semi-online algorithm shown as \Cref{alg:semi}, which uses an augmentation factor of $3$ and acts as the $\Oracle_3$ in our framework.

\begin{algorithm}[ht]
  \caption{The Semi-Online Algorithm}
  \label{alg:semi}
  \begin{algorithmic}
  \Procedure{{\sf SemiOnline}}{$J$: input job set, $T$: length of rent}
  \State $J', I \gets \varnothing$ \Comment{$I$ is a multiset for rents.}
  \State $\tau_j = \max \{r_j, d_j - T\}$ for all $j$. 
  \For {$j \in J$ in non-decreasing order of $\tau_j$}
  \State $J' \gets J' \cup \{j\}$
  \If {$\textsf{EDF}(J', I)$ fails}
  \State $I \gets I \cup \left\{\left[\tau_j - T, \tau_j +2T\right)\right\}$ \Comment{A rent that starts at $\tau_j-T$ with length $3T$.}
  \EndIf
  \EndFor
  \State \Return $I$
  \EndProcedure
  \end{algorithmic}
\end{algorithm}

We call Algorithm \ref{alg:semi} semi-online, because the enumerating order is exactly the same as how $J_t$ increases in the oracle-based framework. Thus, if we have some new jobs with $r_j=t$ or $d_j -T = t$ when the online time moves from $t-1$ to $t$, the only possible difference between $\Semi(J_{t-1},T)$ and $\Semi(J_{t},T)$ is some rent intervals of $[t-T,t+2T)$. This observation could be formalized into the following properties of the semi-online algorithm. 


\begin{lemma} [Strong Monotonicity]
    \label{lem:semionlineproperty}
    We have the following two properties for \Cref{alg:semi}.
    
    \begin{enumerate}
        \item For any $J_{t_1}$ and $J_{t_2}$ where $t_1\leq t_2$, we have 
        $
            \Semi(J_{t_1},T) \subseteq \Semi(J_{t_2},T).
        $
        \item $\Semi(J_t,T) \setminus \Semi(J_{t-1},T)$ is a multiset of a fixed rent interval $[t-T,t+2T)$.
    \end{enumerate}
\end{lemma}
\begin{proof}
    Intuitively, the reason behind the lemma is that the order of $\tau_j$ is the same as the order in which we insert jobs into $J_t$ as $t$ increases. Formally speaking, compare $J_{t_1}$ and $J_{t_2}$ and consider a job $j \in J_{t_2} \setminus J_{t_1}$. By definition, we have that $\tau_j \geq \displaystyle\max_{j' \in J_{t_1}} \tau_{j'}$. Therefore, \Cref{alg:semi} first enumerates the jobs in $J_{t_1}$ and then the jobs in $J_{t_2} \setminus J_{t_1}$, which concludes the first property immediately. For the second property, the reason is that $J_{t} \setminus J_{t-1}$ is a set of jobs with $r_j=t$ or $d_j-T=t$. In other words, we enumerate them after jobs in $J_{t-1}$. Thus, if $I$ continues to grow when we enumerate them, the new interval must be $[t-T,t+2T)$.
\end{proof}


The strong monotonicity in \Cref{lem:semionlineproperty} suffices to show the weak monotonicity in property 2) of \Cref{lem:oracle}. On the other hand, these two properties are also used in the proof of property 3) later. It remains to prove property 1) by bounding the cardinality of the semi-online algorithm's solution.

\begin{lemma}
    \label{lem:semicost}
    $|\Semi(J,T)| \leq \OPT(J,T)$, and $\Semi(J,T)$ is feasible for $J$.
\end{lemma}

Before proving Lemma \ref{lem:semicost}, we introduce $\CoOPT$ so that we can better understand the solution structure.

\begin{definition}[Optimal complement solution]
    For a job set $J$, rent length $T$, and a given rent set $I$ (which may not be length $T$), the optimal complement solution of $I$, denoted as $\CoOPT (I)$, is defined as a rent set of length $T$ with minimum cardinality such that $I \cup \CoOPT (I) $ is feasible for $J$.
\end{definition}
\begin{fact}
\label{fac:coopt}
$\CoOPT (\varnothing) = \OPT, \CoOPT (\OPT) = \varnothing.$
\end{fact}

Consider a rent set $I$ that is infeasible for $J$. Below, we state the main property of \CoOPT.

\begin{lemma}
\label{lem:coopt}
 Let $d$ be the fail time of $\EDF(J,I)$. There exists a $\CoOPT(I)$ such that there is a rent interval $[s,s+T) \in \CoOPT(I)$ that satisfies:
$
d - T \leq s < d.
$
\end{lemma}

\begin{proof}
We use $\CoOPT$ as a shorthand for $\CoOPT(I)$ in this proof. Let $[s,s+T)$ be the earliest interval in \CoOPT. 

First, we show that $s < d$. Suppose, for a contradiction, that $s \geq d$. Let $j$ be the job that fails in $\EDF(J,I)$, where $J$ is a fixed given job set. Then, $j$ has no more active units in $\CoOPT$, since all rent intervals start at or after $d$. But this contradicts the definition of $\CoOPT$, which is a feasible rent set for $J$.

Second, we show that $s \geq d - T$ can be true. If this is not true, we construct a new rent set 
$
{\CoOPT'} = \CoOPT \setminus \left\{[s,s+T)\right\} \cup \left\{[d-T,d)\right\}.
$
That is, we replace the rent interval at $s$ with another one at $d-T$. We claim that $\CoOPT' \cup I$ is also feasible for $J$. This means that $\CoOPT'$ is also a feasible $\CoOPT$, and $[s' = d-T,T)$ is a feasible rent interval that satisfies our desired condition.

To prove our claim, we fix an arbitrary choice of $r^*\leq d^* \in \mathbb{N}$, and we verify the condition in \Cref{lem:fea}, i.e., 
$
A_{I \cup \CoOPT'}(r^*,d^*) \geq |J(r^*,d^*)|. 
$

We consider two cases: 
\begin{itemize}
    \item Case 1: $d^* < d$. If the condition does not hold, we have $A_{I}(r^*,d^*) \leq A_{I \cup \CoOPT'}(r^*,d^*) < J'(r^*,d^*)$, which means that $\EDF(I,J')$ must fail no later than $d^*$ since the active units are not enough beforehand. This contradicts the definition of $d$.
    
    \item Case 2: $d^* \geq d$. The only difference between $\CoOPT$ and $\CoOPT'$ is the contribution of active units by $[s,s+T)$ and $[d-T,d)$. We prove that $[d-T,d)$ must provide at least as many active units as $[s,s+T)$ does. 
    \begin{figure}[ht]
    \centering
    \includegraphics[width=110mm]{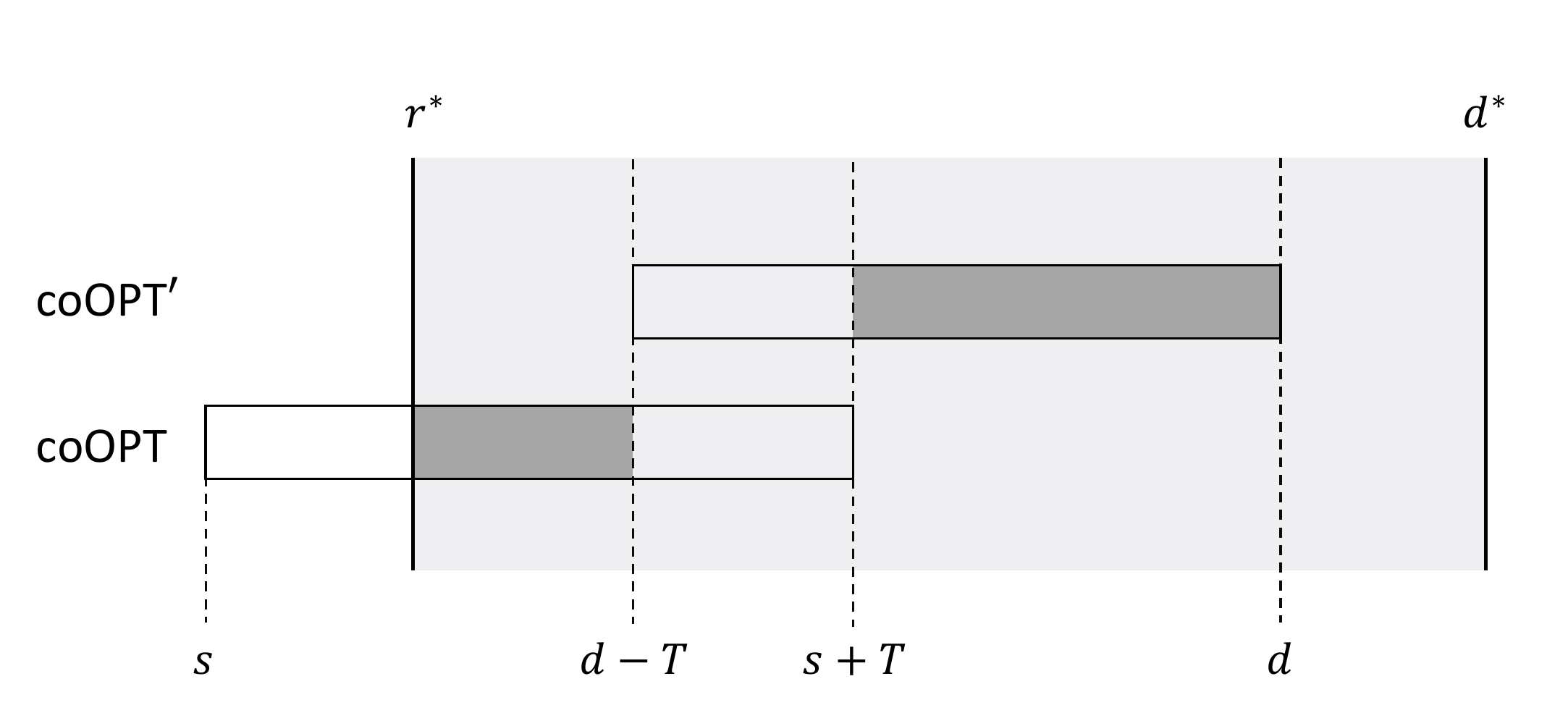}
    \caption{\CoOPT' contributes at least as many active units as \CoOPT on $[r^*,d^*)$ when $d^*\geq d$.}
    \label{fig:semiproof}
\end{figure}
    Referring to \Cref{fig:semiproof}, we see that $[d-T,d)$ has more active units than $[s,s+T)$ in $[s+T,d)$, and vice versa in $[s,d-T)$. Since $d^*\geq d$, we need $r^* \leq d-T$ to reach the advantage area of $\CoOPT$; however, $[r^*,d^*)$ then covers the whole part of $[d-T,d)$. This implies that the total contribution of $\CoOPT$ never exceeds that of $\CoOPT'$. 
\end{itemize}

The discussion concludes the claim.

\end{proof}

\begin{lemma}
    \label{lem:semifeasiblity}
    At the end of each iteration of $j$, $I$ is feasible for $J'$.  
\end{lemma}
\begin{proof}
    We prove it by induction. In the base case, $I$ is feasible for $J'$ when they are $\varnothing$. Then, assume the lemma is true after the $(j-1)$-th iteration. At the $j$-th iteration, we add a job $j$ to $J'$. This means that $\forall r^* \leq r_j,~ d^* \geq d_j$, $|J'(r^*,d^*)|$ will increase by one. If $\EDF(J',I)$ is already feasible, we are done. Otherwise, the algorithm will employ a new $3T$ length rent interval $[\tau_j - T, \tau_j + 2T)$. Notice that 
    $d^* \geq d_j \geq \tau_j$. Every $A_{I}(r^*,d^*)$ also increases by at least one. Thus, we still have $A_{I}(r^*,d^*)\geq |J'(r^*,d^*)|$ after we employ $[\tau_j-T,\tau_j+2T)$ (at the end of the $j$-th iteration). 
\end{proof}
\begin{corollary}
$\Semi(J,T)$ is feasible for $J$.
\end{corollary}
\begin{lemma}
\label{lem:semimain}
    In the $j$-th iteration, if $\EDF(J',I)$ is infeasible in Algorithm \ref{alg:semi}, before we rent $[t-T,t+2T)$, we have
    $$
    |\CoOPT( I \cup \left\{[t-T,t+2T)\right\})| \leq |\CoOPT(I)| - 1.
    $$
\end{lemma}

\begin{proof}
    By the condition of the lemma and \Cref{lem:semifeasiblity}, we have $\EDF(J', I)$ is infeasible while $\EDF(J' \setminus \{j\}, I)$ is feasible. By the enumerating order, all the jobs $j'$ in $J'$ must satisfy $\max\{r_j,d_{j'}-T\} \leq \tau_j$. Therefore, $\EDF(J', I)$ must fail at a deadline $d \leq \tau_j+T$. On the other hand, for all $j' \in J$ such that $d_{j'} < \tau_j$, we must have $j' \in J' \setminus \{ j \}$, also because of the enumerating order. Thus, we can show that $d \geq \tau_j$. Otherwise, $J' \setminus \{ j\}$ would be infeasible for $I$, which is a contradiction. In conclusion, we show that the failure time $d$ of $\EDF(J',I)$ satisfies
    $
    \tau_j \leq d < \tau_j + T.
    $
    
	Finally, by \Cref{lem:coopt}, there exists a \CoOPT with rent interval $[s,s+T]$ such that $d-T \leq s < d$. It implies that $\tau_j - T < s < \tau_j + T$. Therefore $[s, s+T)$ is always a subset of $[\tau_j - T, \tau_j + 2T)$. We have
	$$
	|\CoOPT( I \cup \left\{[t-T,t+2T)\right\})| \leq |\CoOPT( I \cup \{ s \})| = |\CoOPT(I)| - 1.
	$$
    
    
\end{proof}

\begin{proof}[Proof of \Cref{lem:semicost}]
    Recall \Cref{fac:coopt} that $\CoOPT(\varnothing) = \OPT$. It follows that $\sf{SemiOnline}$ rents at most $\OPT$ times as a corollary of \Cref{lem:semimain}.
\end{proof}




%% file: onlinealg.tex


\section{The 6-competitive Online Algorithm}
\label{sec:finalalgorithm}


It remains to define the rent scheme for each \Batch in the framework. For completeness, we formally describe the algorithm in \Cref{alg:online}.
\begin{algorithm}
  \caption{The Online algorithm}
  \label{alg:online}
  \begin{algorithmic}
  \Procedure{{\sf OnlineRent}}{$t$: time, $J$: known jobs, $T$: length of rent}
  \State $\Delta = |{\sf SemiOnline}(J_t, T)| - |{\sf SemiOnline}(J_{t - 1}, T)|$
  \State Perform $\Delta$ \Batch at time $t$, each consists of $6$ machines: $4$ at $t$ and $2$ at $t + T$.
  \State {\bf schedule} jobs at $t$ following $\EDF(J,I)$, where $I$ is the current online rent set.  
  \EndProcedure
  \end{algorithmic}
\end{algorithm}

Next, we prove the property 3) of \Cref{lem:oracle} by our design of \Batch, i.e., to show \Cref{alg:online} is feasible for the total job set $\calJ$. Combining with the property 1) and 2) by the semi-online algorithm, we can conclude our online algorithm is $6$-compeititve as claimed in \Cref{thm:main}. 

First, we introduce an obvious relationship between online and semi-online algorithms.
\begin{fact}
    \label{lem:onlineofflinemapping}
    At every moment $t$, there always exists a bijection from one semi-online rent batch to one online rent batch, such that both batches are at the same time: $4\times[t, t + T) + 2 \times [t + T, t + 2T) \text{ in \textup{\sf Online}} \mapsto [t - T, t + 2T) \text{ in \textup{\sf SemiOnline}}$.
\end{fact}
\begin{proof}
    This fact is directly implied by the second property of \Cref{lem:semionlineproperty}. Whenever $\Semi(J_t,T)$ increases from $\Semi(J_{t-1},T)$ by some rent intervals of $[t - T, t + 2T)$, the new batches must be $4\times[t, t + T) + 2 \times [t + T, t + 2T)$. Each of them can correspond to one $[t - T, t + 2T)$.
\end{proof}

We prove the feasibility by showing that the active units provided by the online algorithm are always enough for the possible jobs inside any possible interval $[r^*,d^*)$.  



\begin{lemma}
    \label{thm:online:lemma1}
    Let $I$ be the rent sets made by our algorithm. We have 
    $ \forall r^* \leq d^* \in \mathbb N,~ A(r^*,d^*)\geq |\calJ(r^*,d^*)|$, where we use $A$ to denote $A_I$ for simplicity.
\end{lemma}

We remark that \Cref{thm:online:lemma1} provides the necessary information to prove the correctness via the feasibility lemma (\Cref{lem:fea}). It only remains to complete the proof of \Cref{thm:online:lemma1}. 

\subsection{Proof of Lemma \ref{thm:online:lemma1}}
Let us fix an arbitrary range $r^*\leq d^*$, and discuss $A(r^*,d^*)$ and $\calJ(r^*,d^*)$ separately.
First, we discuss $A(r^*,d^*)$.  The behavior of the online algorithm can be represented by a set of online rent batches. Note that only those rent batches that start in $\left(r^*-2T,d^*\right)$ can provide active units inside $\left[r^*,d^*\right)$. Therefore, we only discuss a subset $B$ of all rent batches with a start time in $\left(r^*-2T,d^*\right)$. Each $b$ means a \Batch made by \Cref{alg:online}. We use $t(b)$ to mean its decision time. That is, a batch $b$ contains 4 rent intervals of $[t(b),t(b)+T)$, and 2 of $[t(b)+T,t(b)+2T)$.  

We partition the time interval $[r^*,d^*)$ by several critical time points. The first time point is 
$
\theta_1=\max\left\{\left(\frac{r^*+d^*}{2}\right), d^*-T\right\}.
$ 
It means that $\theta_1 = d^*-T$ when $d^*-r^*\geq 2T$ and $\theta_1=\left(\frac{r^*+d^*}{2}\right)$ when $d^*-r^* < 2T$. Remark that in both cases, $\theta_1 \geq d^*-T$. Then we recursively define $\theta_{i + 1}=\left(\frac{\theta_{i} + d^*}{2}\right)$ for every $i\geq 1$ until $\lfloor \theta_i \rfloor = d^*-1$. Besides, we let $\theta_0=r^* - 1$.
For $i\geq 1$, we call $\left[\left\lfloor \theta_{i - 1} \right\rfloor + 1, \left\lfloor \theta_{i} \right\rfloor\right]$ the $i$-th sub-interval, which is the minimal sub-interval of $(\theta_{i - 1}, \theta_{i}]$ that contains all integers in it. We define $B_i\subseteq B$ as the set of rent batches starting in the $i$-th sub-interval. Moreover, we let $B_0$ be the set of rent batches starting in $(r^*-2T, r^*]$.





For each $b\in B$, recall that $t(b)$ is the time it was allocated, and we let $\lambda_i(b)$ be the length of the intersecting interval of $[t(b),t(b)+2T)$ and $\left[\left\lfloor \theta_{i - 1} \right\rfloor + 1, d^*\right)$. Note that $\lambda_1(b)$ represents the intersecting interval with the whole $\left[r^*, d^*\right)$. Let $L = \min\{2T, d^* - r^*\} = 2 (d^* - \theta_1)$ denote the maximum possible length in $\lambda_1$, and by our partition method. It follows the property of the length of sub-intervals by our partition.

\begin{lemma}[Partition length property]
    \label{lem:L}
    For any batch $b \in B_i$ where $i \geq 1$, the intersection of $\left[t(b), t(b) + 2T\right)$ and $[\left\lfloor \theta_{i - 1} \right\rfloor + 1, d^*)$ satisfies:
    $
    2^{-i} \cdot L \leq \lambda_i(b)  \leq 2^{1-i}\cdot L.
    $
\end{lemma}
\begin{proof}
    For $i = 1$, $\lambda_1(b) = \min \{2T, d^* - t(b)\}$. By definition, $r^*\leq t(b) \leq \lfloor \theta_1 \rfloor $, hence
    $$ L / 2 = \min \{T, d^* - \theta_1\} \leq \min \{2T, d^* - t(b)\} \leq \min \{2T, d^* - r^*\} = L.$$

    For $i \geq 2$, by definition $d^* - \theta_i = 2^{-i} \cdot L$. Because $t(b) \in \left[\left\lfloor \theta_{i - 1} \right\rfloor + 1, \left\lfloor \theta_{i} \right\rfloor\right] \subseteq (\theta_{i-1},\theta_i]$, we have $ 2^{-i} \cdot L \leq  \lambda_i(b)  \leq \min\{2T,~2^{1-i}\cdot L\}$. Since $2 \cdot L \leq 2T$, we conclude the lemma.  
\end{proof}



Then, we present the lemmas for a lower bound of active units and an upper bound of job numbers. 

\begin{lemma}[Lower bound of active units]
    \label{lem:activeunits}
    $$
    \begin{aligned}
    A(r^*,d^*)\geq &\sum_{b\in B_0} \left(2 \lambda_1(b) + 2 \max\{\lambda_1(b) - T, 0\}\right) \\
    &+ \sum_{b \in B_1} \left(2\lambda_1(b) + 2\min\{\lambda_1(b), T\}\right) \\
    &+ \sum_{i \geq 2} 4 \cdot(2 ^ {-i} \cdot L) \cdot |B_i|.
    \end{aligned}
    $$
\end{lemma}
\begin{proof}
    \begin{figure}
        \centering
        \includegraphics[width=120mm]{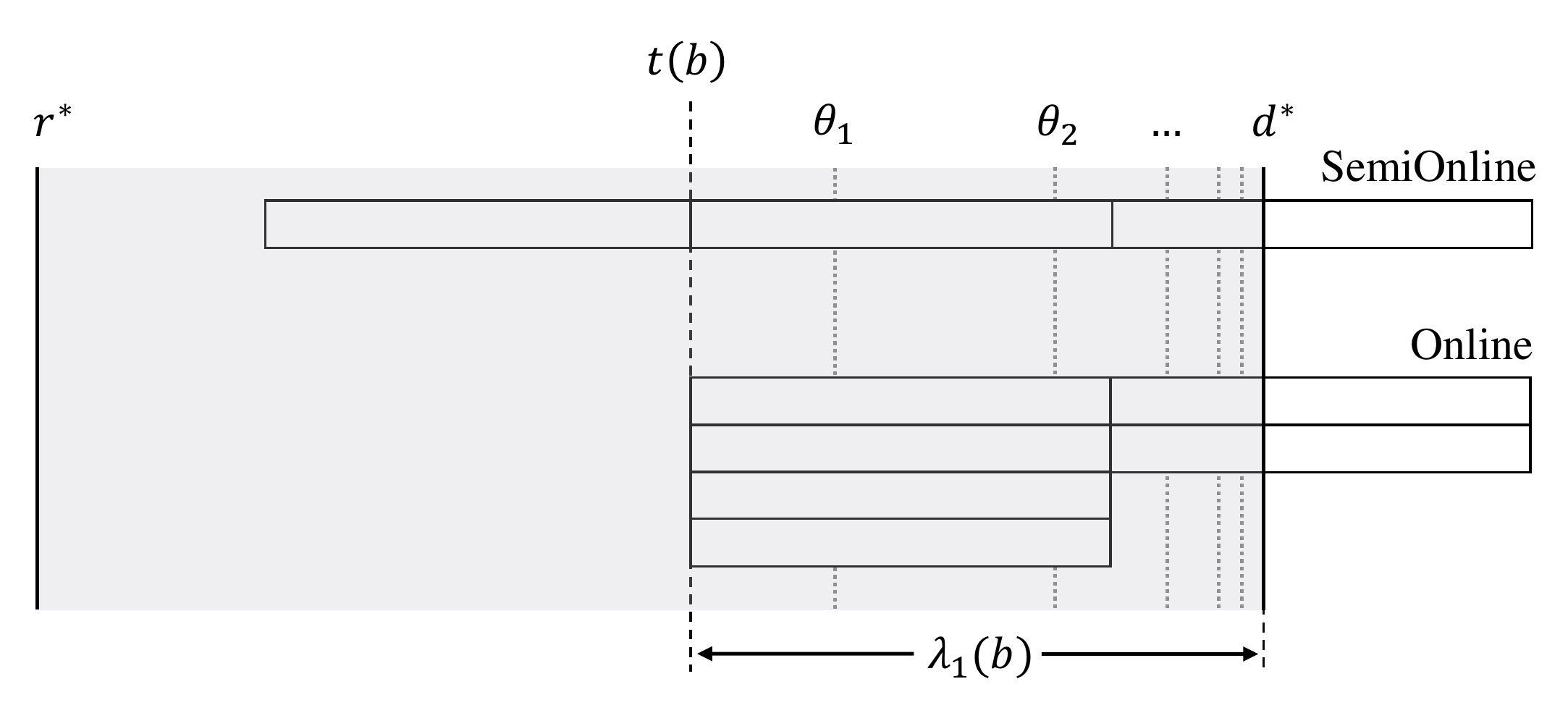}
        \caption{An example for $b \in B_1$ when $d^*-r^* > 2T$. The shaded area is the considered $[r^*, d^*)$, and the online rent batch provides $2\lambda_1(b) + 2\min\{\lambda_1(b), T\}$ active units to it.} 
        \label{fig:lambda}
    \end{figure}
    Three terms on the RHS of the inequality are counting of $B_0, B_1$ and $B_{\geq 2}$, where the first two are straightforward counting as shown in \Cref{fig:lambda}, and the last term was scaled down a bit by \Cref{lem:L}: 
    $$
        4  \sum_{i \geq 2} \sum_{b \in B_i} \lambda_1(b)
        = 4  \sum_{i \geq 2} \sum_{b \in B_i} \lambda_i(b)
        \geq \sum_{i \geq 2} 4 \cdot (2 ^ {-i} \cdot L) \cdot |B_i|.
    $$
\end{proof}
Let $\mathcal J_i$ be the job set with deadline at most $d^*$ and released in the $i$-th subinterval:
$$\mathcal J_i=\left\{j\in \calJ \mid \lfloor \theta_{i-1}\rfloor + 1 \leq r_j \leq \lfloor \theta_i \rfloor,\ d_j\leq d^*\right\}.$$
We provide an upper bound of $\mathcal J_i$ by the performance of our algorithm. 
\begin{lemma}
\label{lem:semi_job_number}
    Let $I_t$ be the semi-online batches allocated at or before time $t$. 
    We have that $\calJ_i$ is no more than the active units after $\lfloor \theta_{i - 1} \rfloor + 1$ provided by \Semi at $\lfloor \theta_{i} \rfloor$.
    i.e.,
    $$\mathcal J_i \leq A_{I_{\lfloor \theta_{i} \rfloor}} (\lfloor \theta_{i - 1} \rfloor + 1 , d^*).$$
\end{lemma}

\begin{proof}
    Let us observe the time point $t=\lfloor \theta_{i} \rfloor$. $\Semi(J_t,T)$ reports $I_t$ at this time. By \Cref{lem:semifeasiblity}, $I_t$ is feasible for $J_t$. By the definition of $\theta_i$, we prove $t \geq d^* - T$ because $\theta_1 \geq d^* - T$. Thus, all jobs with deadlines at most $d^*$ are already in $J_t$, and combining with the feasibility of $I_t$, we have:
    $$
        \calJ_i = J_t(\lfloor \theta_{i - 1} \rfloor + 1, d^*) \leq A_t(\lfloor \theta_{i - 1} \rfloor + 1, d^*). 
    $$
    The lemma then concludes because we define $t=\lfloor \theta_{i} \rfloor$.
\end{proof}

\begin{lemma}[Upper bound of job numbers]
\label{lem:job_number_general}
We have two different upper bounds for $\mathcal J_i$:

\begin{itemize}
    \item $i = 1$: $\quad \displaystyle |\mathcal J_1| \leq \sum_{b\in B_0} \lambda_1(b) + \sum_{b \in B_1} \left(\lambda_1(b) + L/2\right).$
    \item $i \geq 2$: $\quad \displaystyle |\mathcal J_i| \leq \sum_{b\in B_0} \lambda_i(b) + 2 \cdot (2 ^ {-i} \cdot L) \cdot \sum_{j=1}^{i} |B_j|.$
\end{itemize}

\end{lemma}

\begin{proof}

In this proof, we apply \Cref{lem:semi_job_number} and count the number of active units after $\theta_{i - 1}$ by $I_{\lfloor \theta_{i} \rfloor}$. Note that by \Cref{lem:onlineofflinemapping}, each $3T$ length rent interval in $I_{\lfloor \theta_{i} \rfloor}$ corresponds to an online \Batch. 

First, let us consider the case $i=1$. $I_{\lfloor \theta_{i} \rfloor}$ corresponds to the online batches with $t(b) \leq I_{\lfloor \theta_{i} \rfloor}$. Note that the ending time of $b$ and its corresponding semi-online rent interval are both $t(b)+2T$. Thus, $A_{I_{\lfloor \theta_{i} \rfloor}} (\lfloor \theta_{i - 1} \rfloor + 1 , d^*)$ corresponds to $B_0$ and $B_1$.


For $i = 1$, we keep the $B_0$ straightforward and calculate the upper bound of active units provided by the corresponding semi-online batch for each batch $b$ in $B_1$. Note that the semi-online rent set spans $[t(b) - T, t(b) + 2T)$. We split the $3T$ interval into first $T$ and last $2T$: the latter could be upper bounded by $\lambda_1(b)$ using \Cref{lem:L}, and we could find that the former is at most $L/2$:
    \begin{itemize}
        \item When $d^* - r^* < 2T$, $t(b) - r^* < (d^* - r^*) / 2$, then $\min\left\{T, t(b) - r^*\right\} < L / 2$.
        \item When $d^* - r^* \geq 2T$, $L = 2T$, and then $\min\left\{T, t(b) - r^*\right\} \leq T = L / 2$.
\end{itemize}
So we can conclude that
    $$|\mathcal J_1| \leq \sum_{b\in B_0} \lambda_1(b) + \sum_{b \in B_1} \left(\lambda_1(b) + L/2\right).$$

For the case $i \geq 2$, all jobs released in sub-interval $i$ can be allocated at most $[\theta_{i-1} + 1, d^*)$, and by \Cref{lem:L} each semi-online batch covers at most $2 \cdot (2^{-i} \cdot L)$. Like $i=1$, the inequality follows direct counting on $\cup_{j = 0}^{i}B_j$.
\end{proof}

Thus far, we are ready to prove \Cref{thm:online:lemma1} by a charging argument. 

\begin{proof}[Proof of \Cref{thm:online:lemma1}]
 Recall that in \Cref{lem:activeunits} we have 
\begin{equation}
\label{eqn:ALower}
\begin{aligned}
A(r^*,d^*)\geq 
    &\sum_{b\in B_0} \left(2 \max\{\lambda_1(b) - T, 0\} + 2 \lambda_1(b)\right) 
    + \sum_{b \in B_1} \left(2\lambda_1(b) + 2\min\{T, \lambda_1(b)\}\right) \\
    &+ \sum_{i \geq 2} 4 \cdot(2 ^ {-i} \cdot L) \cdot |B_i|.
\end{aligned}
\end{equation}
Also, by \Cref{lem:job_number_general},
\begin{equation}
    \label{eqn:JUpper}
    \calJ(r^*,d^*) \leq
    \sum_{b\in B_0} \lambda_1(b)
    + \sum_{b \in B_1} \left(\lambda_1(b) + L / 2 \right)
    + \sum_{i \geq 2} \left(\sum_{b\in B_0} \lambda_i(b) + 2 \cdot (2 ^ {-i} \cdot L) \cdot \sum_{j=1}^{i} |B_j|\right).
\end{equation}

Using these two inequalities, we charge the upper bound of $\calJ(r^*, d^*)$ and the lower bound of $A(r^*, d^*)$ to each $b \in B$. We prove that for each $b$, the charged amount of $\calJ(r^*, d^*)$'s upper bound is at most $A(r^*, d^*)$'s lower bound. 

First, for $b \in B_0$, the contribution of $b$ to the lower bound of $A(r^*, d^*)$ (i.e., RHS of \Cref{eqn:ALower}) is: 
$
2 \max\{\lambda_1(b) - T, 0\} + 2 \lambda_1(b).
$
The contribution to the upper bound of $\calJ(r^*,d^*)$ (i.e., RHS of \Cref{eqn:JUpper}) is: 
$
\lambda_1(b) + \sum_{i \geq 2} \lambda_i(b).
$
Note that $b \in B_0$ only contributes to a prefix of $[r, d)$, it is easy to see that $\lambda_i(b) \leq 2 ^ {1 - i} \lambda_1(b)$, and thus
$$
\lambda_1(b) + \sum_{i \geq 2} \lambda_i(b) \leq \lambda_1(b) + \sum_{i \geq 2} 2^{1-i} \lambda_1(b) < 2\lambda_1(b).
$$
We are done for $b\in B_0$. 

Second, for $b \in B_1$, the contribution of $b$ to the lower bound of $A(r^*, d^*)$ (i.e., RHS of \Cref{eqn:ALower}) is: 
$
2\lambda_1(b) + 2\min\{T, \lambda_1(b)\}.
$
The contribution to the upper bound of $\calJ(r^*,d^*)$ (i.e., RHS of \Cref{eqn:JUpper}) is: 
$
\lambda_1(b) + L/2 + \sum_{i \geq 2} 2 \cdot (2 ^ {-i} \cdot L).
$
Then, we have 
\begin{align*}
    \lambda_1(b) + L/2 + \sum_{i \geq 2} 2 \cdot (2 ^ {-i} \cdot L) 
    &< \lambda_1(b) + L/2 + 2\cdot \left(2 ^ {-1} \cdot L\right) \\
    &\leq \lambda_1(b) + L/2 + 2 \cdot \min\left\{T, \lambda_1(b)\right\} \\
    &\leq 2\lambda_1(b) + 2\min\{T, \lambda_1(b)\}.
\end{align*}
The last inequality holds by \Cref{lem:L}. Therefore, we are done for $b\in B_1$. 

Finally, for $b \in B_{i\geq 2}$, the contribution of $b$ to the lower bound of $A(r^*, d^*)$ (i.e., RHS of \Cref{eqn:ALower}) is: 
$
4 \cdot(2 ^ {-i} \cdot L).
$
The contribution to the upper bound of $\calJ(r^*,d^*)$ (i.e., RHS of \Cref{eqn:JUpper}) is: 
$
\sum_{i' \geq i} 2 \cdot (2 ^ {-i'} \cdot L).
$
We are done because 
$
\sum_{i' \geq i} 2^{-i'} < 2^{1 - i}.
$
Summing up three parts, we have proved that $A(r^*, d^*) \geq \mathcal J(r^*, d^*)$.

\end{proof}

\subsection{A Remark on Running Time.} 

We only need to recalculate ${\sf SemiOnline}$ if the job set gets updated, and the procedure of reconstructing the rent set in ${\sf SemiOnline}$ can be maintained incrementally. Thus, it is possible to implement the algorithm calling $\EDF$ at most $n + \OPT \leq 2n$ times, and hence achieve a worst case guarantee of $O(n ^ 2 \log n)$. Also, note that a job can influence the calculation of ${\sf SemiOnline}$ for at most $O(T)$ time units, so the ${\sf SemiOnline}$ can also update in $O(n w \log w)$ if there are at most $O(w)$ jobs within any interval of length $O(T)$. Therefore, the online algorithm is efficient in terms of worst case guarantee and also average online updating.

%% file: lambda_full.tex
\section{Online Rent Minimization with Delay}

In the version with delay, we are also given an online released job set $\mathcal{J}$ and a rent length $T$. We aim to minimize the number of rents needed to process all jobs. The notations are the same as in the Online Rent Minimization problem. Moreover, we are given a nonnegative integer $\lambda$ as the delay parameter. It means that if we rent a machine at time $t$, we will have an active machine at $[t+\lambda,t+\lambda +T)$. Note that it is impossible to serve an unknown emergency job with $d_j-r_j \leq \lambda$ online; following Chen and Zhang~\cite{chen_online_2022}, we require that the active time $d_j - r_j$ is at least $\lambda + 1$. 


We use the following reduction lemma and our $6$-competitive no-delay algorithm as a black box to prove \Cref{thm:timecritical}. Chen and Zhang~\cite{chen_online_2022} also mention this approach.

\begin{lemma}[Reduction]
    \label{lambda}
    If $\ALG(\mathcal{J}) \leq \Gamma \cdot \OPT(\mathcal{J})$ for every job set $\mathcal{J}$, we have an algorithm $\ALG_\lambda$ that guarantees $$\ALG_\lambda(\mathcal{J}, \lambda) \leq \Gamma \cdot (\lambda + 1) \cdot \OPT(\mathcal{J}),$$ if $\forall j \in \mathcal{J}$, $d_j - r_j \geq \lambda + 1$.
\end{lemma}

\subparagraph*{Construction of $\mathcal{J}_{\lambda = 0}$ and $\ALG_\lambda$.}
$\mathcal{J}_{\lambda = 0}$ is a modified job set by $\mathcal{J}$. For any job $j$ in $\mathcal{J}$, we construct a corresponding $j'$ in $\mathcal{J}_{\lambda = 0}$ such that $r_{j'} = r_j$ and $d_{j'}=d_j - \lambda$. Then, we construct $\ALG_\lambda$. At any time $t$, $\ALG_\lambda(\mathcal{J})$ will make the rent decision as the decisions of $\ALG(\mathcal{J}_{\lambda = 0})$ made at $t$. That is, if there is a rent interval $[t,t+T)$ in $\ALG(\mathcal{J}_{\lambda = 0})$, we will have a rent interval $[t+\lambda,t+T + \lambda)$ in $\ALG_\lambda(\mathcal{J})$. Because we do not change the releasing time, $\ALG_\lambda$ is also an online algorithm. Clearly, the cost of $\ALG_\lambda(\mathcal{J})$ is the same as that of $\ALG(\mathcal{J}_{\lambda = 0})$.
\begin{lemma}
        \label{lem:delaycost}
    $\ALG_\lambda(\mathcal{J},\lambda) \leq \Gamma (\lambda + 1) \cdot \OPT(\mathcal{J})$.
\end{lemma}
\begin{proof}
    We have  $\ALG_\lambda(\mathcal{J},\lambda) = \ALG(\mathcal{J}_{\lambda = 0})$ by the construction. Then, let us refer to the lemma for $\OPT$ in Chen and Zhang~\cite{chen_online_2022}:
    \begin{lemma}[\cite{chen_online_2022}]
    \label{lem:ref}
        $\OPT(\mathcal{J}_{\lambda = 0}) \leq (\lambda + 1) \cdot \OPT(\mathcal{J}).$
    \end{lemma}
    By Lemma \ref{lem:ref}, $$\ALG_\lambda(\mathcal{J},\lambda) = \ALG(\mathcal{J}_{\lambda = 0}) \leq \Gamma \cdot \OPT(\mathcal{J}_{\lambda = 0}) \leq \Gamma (\lambda + 1) \cdot \OPT(\mathcal{J}).$$
\end{proof}

\begin{lemma}
    \label{lem:delayfeasible}
    $\ALG_\lambda(\mathcal{J})$ is feasible to process all jobs in $\mathcal{J}$.
\end{lemma}
\begin{proof}
    We apply \Cref{lem:fea} to prove the feasibility. Let $I$ be the rent set of $\ALG(\mathcal{J}_{\lambda = 0})$. Given that $\ALG(\mathcal{J}_{\lambda = 0})$ is feasible, we know that $\forall r^* < d^*$, $\mathcal{J}_{\lambda = 0}(r^*,d^*) \leq A_{I}(r^*,d^*)$. Next, we consider $\ALG_\lambda(\mathcal{J})$ and let its rent set be $I'$.  $\forall r^*<d^*$ such that $d^*-r^* \leq \lambda$, $\mathcal{J}(r^*,d^*)=0$ by definition. For those $d^* - r^* \geq \lambda+1$, 
    $$
    \mathcal{J}(r^*,d^*) = \mathcal{J}_{\lambda = 0}(r^*,d^* - \lambda) \leq A_{I}(r^*,d^* - \lambda).
    $$ 
    By definition, the active unit in $A_{I}(r^*,d^* - \lambda)$ will be postponed to $A_{I'}(r^* + \lambda, d^*)$. Thus, we have
    $$
    \mathcal{J}(r^*,d^*) \leq A_{I}(r^*,d^* - \lambda) = A_{I'}(r^* + \lambda, d^*) \leq A_{I'}(r^*, d^*).
    $$
    Therefore, we can use \Cref{lem:fea} to show the feasibility of $I'$ in both cases. 
\end{proof}

We now can conclude the correctness of \Cref{lambda} from \Cref{lem:delaycost} and \Cref{lem:delayfeasible},  \Cref{thm:timecritical} also follows directly. 

\timecrit*

%% file: conclusion.tex
\section{Conclusion and Future Work}
In conclusion, our main contribution is a $6$-competitive algorithm for the Online Rent Minimization problem under unit-size jobs, which follows the oracle-based framework. 

Since the Online Rent Minimization problem is a generalization of the Online Machine Minimization problem, where we have an optimal $e$-competitive algorithm, one major question is: Is the Rent Minimization problem strictly harder than the Machine Minimization problem? 

On the other hand, we are also interested in the power of oracle-based algorithms. Note that the optimal $e$-competitive algorithm for Machine Minimization follows the oracle-based framework. It is interesting to ask: What is the best competitive ratio we can achieve for the Online Rent Minimization problem by using the oracle-based framework? The Semi-Online captures our current understanding of the possible range of optimal solutions, so replacing it with an optimal oracle cannot improve ratio directly by the same argument in the paper. Is it possible to obtain a better ratio with access to an optimal oracle?



